%% file: ms.tex
\pdfoutput=1
\documentclass[a4paper]{article}
\usepackage[T1]{fontenc}
\usepackage{fullpage}

\usepackage{amsmath, amsfonts, amssymb, mathtools} 
\usepackage{amsthm} 
\usepackage{hyperref} 
\usepackage{cleveref} 
\usepackage{booktabs, multirow} 
\usepackage{diagbox} 
\usepackage{lmodern}
\usepackage{authblk}
\usepackage{graphicx}

\input{config} 

\newbox{\myorcidaffilbox}
\sbox{\myorcidaffilbox}{\large\includegraphics[height=1.7ex]{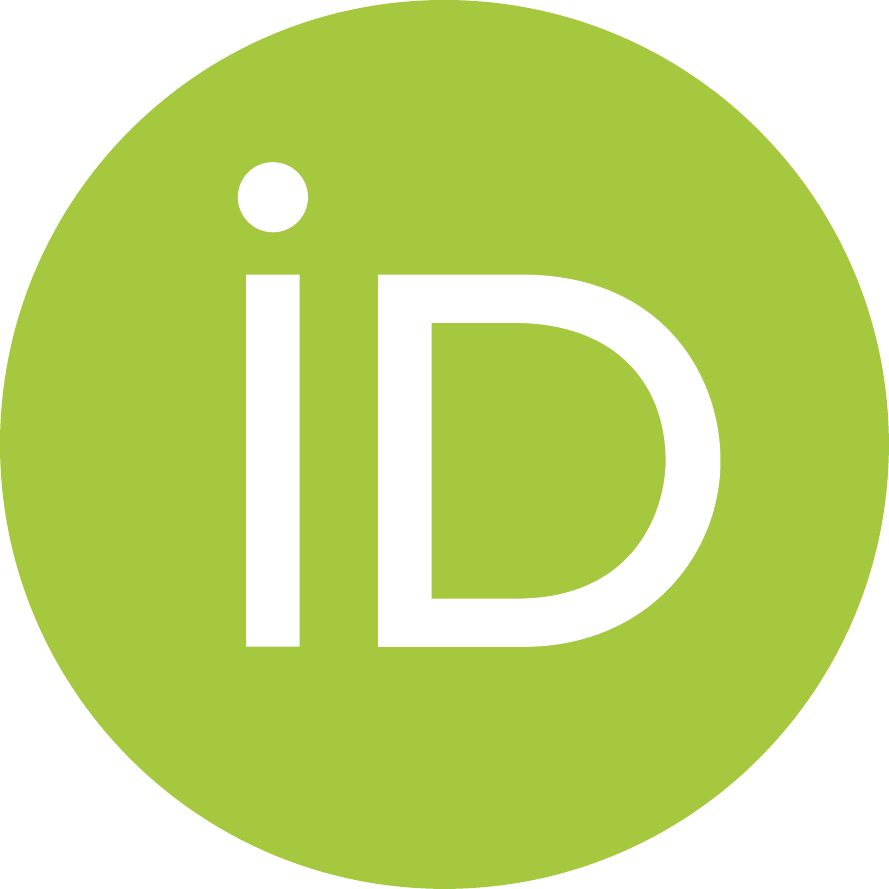}}
\newcommand{\orcidaffil}[1]{%
  \href{https://orcid.org/#1}{\usebox{\myorcidaffilbox}\,#1}}

\title{The Labeled Direct Product Optimally Solves \\ String Problems on Graphs}

\author[1]{Nicola Rizzo}
\author[2]{Alexandru I.\ Tomescu}
\author[3]{Alberto Policriti}
\affil[1]{Department of Computer Science, University of Helsinki, Finland, \texttt{nicola.rizzo@helsinki.fi}, \orcidaffil{0000-0002-2035-6309}}
\affil[2]{Department of Computer Science, University of Helsinki, Finland, \texttt{alexandru.tomescu@helsinki.fi}, \orcidaffil{0000-0002-5747-8350}}
\affil[3]{Department of Mathematics, Computer Science and Physics, University of Udine, Italy, \texttt{alberto.policriti@uniud.it}}
\date{}

\begin{document}

\maketitle

\begin{abstract}
Suffix trees are an important data structure at the core of optimal solutions to many fundamental string problems, such as \emph{exact pattern matching}, \emph{longest common substring}, \emph{matching statistics}, and \emph{longest repeated substring}.
Recent lines of research focused on extending some of these problems to vertex-labeled graphs, although using ad-hoc approaches which in some cases do not generalize to all input graphs.

In the absence of a ubiquitous tool like the suffix tree for labeled graphs, we introduce the labeled direct product of two graphs as a general tool for obtaining optimal algorithms: we obtain conceptually simpler algorithms for the quadratic problems of string matching (\SMLG) and longest common substring (\LCSP) in labeled graphs. Our algorithms are also more efficient, since they run in time linear in the size of the labeled product graph, which may be smaller than quadratic for some inputs, and their run-time is predictable, because the size of the labeled direct product graph can be precomputed efficiently.
We also solve \LCSP\ on graphs containing cycles, which was left as an open problem by Shimohira et al.\ in 2011.

To show the power of the labeled product graph, we also apply it to solve the matching statistics (\MSP) and the longest repeated string (\LRSP) problems in labeled graphs.
Moreover, we show that our (worst-case quadratic) algorithms are also optimal, conditioned on the Orthogonal Vectors Hypothesis. Finally, we complete the complexity picture around \LRSP\ by studying it on undirected graphs. \\[2ex]
String algorithm, Graph algorithm, Motif discovery, Fine-grained complexity
\end{abstract}

\paragraph{Acknowledgments} We are very grateful to Roberto Grossi, for initial discussions on the longest repeated string problem that spurred this line of research, and to Veli M\"akinen and Massimo Equi for their helpful comments and many discussions on the results of this paper.
\paragraph{Funding} This work was partially supported by the European Research Council (ERC) under the European Union's Horizon 2020 research and innovation programme (grant agreement No.~851093, SAFEBIO) and by the Academy of Finland (grants No.~322595, 328877).
\section{Introduction}
\label{sec:introduction}

Motivated by various application domains appearing during the last decades, a significant branch of string algorithm research has focused on extending string problems from texts to more complex objects, such as labeled rooted trees (e.g.~modeling XML documents~\cite{DBLP:conf/focs/FerraginaLMM05}) and labeled graphs (e.g.~modeling pan-genome graphs~\cite{VariationGraph2018,Sch09}). For example, the \emph{string matching in labeled graphs} (\SMLG) problem asks to find an \emph{occurrence} of a given string $S$ inside a labeled graph $G$, that is, a walk of $G$ whose concatenation of vertex labels (\emph{spelling}) is $S$. On rooted trees, \SMLG\ can be solved in linear time~\cite{DBLP:conf/cpm/Akutsu93}, but on general graphs it admits both quadratic-time conditional lower bounds~\cite{BI16,DBLP:conf/icalp/EquiGMT19,EquiMT21,DBLP:conf/sosa/GibneyHT21} and optimal algorithms of matching time complexity~\cite{AmirLL00,RautiainenMarschall17,JZGA19}. 

Despite this active interest in the \SMLG\ problem, the graph extensions of three other fundamental string problems have received none or little attention so far: \emph{longest common substring}, \emph{matching statistics}, \emph{longest repeated substring}. On strings, the former two problems can also be seen as relaxations of the exact string matching problem (for e.g.~handling approximate matching)~\cite{Gus97,DBLP:books/cu/MBCT2015}, and all problems can be seen as basic instances of \emph{{pattern/motif} discovery in strings}~\cite{Par07}. In this paper we consider their natural generalizations to \emph{$\Sigma$-labeled graphs}, namely to tuples $G = (V,E,L)$, with $V$ and $E$ the sets of vertices and edges, respectively, and $L \colon V \to \Sigma$ assigning to each vertex a \emph{label} from $\Sigma$ (the original string problems can be obtained by taking all graphs to be labeled paths).

\begin{problem}[Longest common string problem (\LCSP)]
	Given $G_1$, $G_2$ $\Sigma$-labeled graphs, find a longest string $S$ occurring in $G_1$ and in $G_2$.
\end{problem}

\begin{problem}[Matching statistics problem (\MSP)]
	Given $G_1$, $G_2$ $\Sigma$-labeled graphs, compute for every vertex $v$ of $G_1$ the length $\MS(v)$ of a longest walk of $G_1$ starting at $v$ whose spelling has an occurrence in $G_2$.
\end{problem}

\begin{problem}[Longest repeated string problem (\LRSP)]
	Given a $\Sigma$-labeled graph $G$, find a longest string $S$ having at least two distinct occurrences in $G$.
\end{problem}

\begin{figure}[t]
	\centering%
	\includegraphics[scale=1]{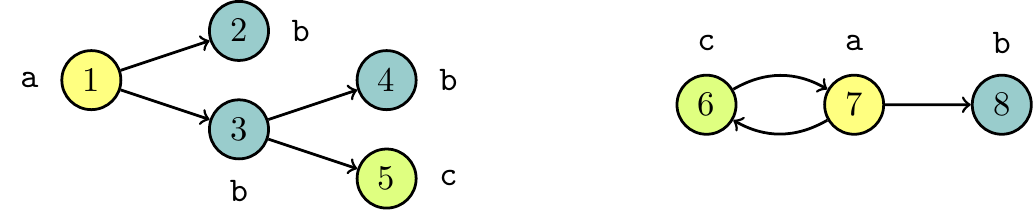}
	\caption{An example of an $\lbrace \mathtt{a}, \mathtt{b}, \mathtt{c} \rbrace$-labeled graph made up of two components: the longest common string between the two components is \texttt{ab}, and the longest repeated string of the whole graph is \texttt{ab} as well, since all longer strings spelled by some walk have exactly one occurrence. Taking the right component as $G_1$ and the left component as $G_2$, we have that $\MS(6) = 1$, $\MS(7) = 2$ and $\MS(8) = 1$.}\label{fig:example}
\end{figure}

When defined on strings, all problems can be solved in linear time and space as basic textbook applications of the suffix tree \cite{Gus97,CR02}, or of the suffix array and the longest common prefix (LCP) array~\cite{Ohl13,DBLP:journals/csur/PuglisiST07}, under the standard assumption to be working with an \emph{integer alphabet}, i.e.\ containing integers from a range that is linear-sized with respect to the input.
On labeled graphs, only \LCSP\ has been considered by Shimohira et al.~\cite{DBLP:conf/stringology/ShimohiraIBT11}. They solved it in time $O( \lvert E_1 \rvert \cdot \lvert E_2 \rvert)$, where $E_1$ and $E_2$ are the edge sets of $G_1$ and $G_2$, respectively, and only if one of the graphs is acyclic, and they left the general case of two cyclic graphs as an open problem~\cite{DBLP:conf/stringology/ShimohiraIBT11}. Moreover, there are no known analogous algorithms for \MSP\ and \LRSP, and no characterization of the possible solutions for \LRSP. Regarding hardness, note that the $O( \lvert E_1 \rvert \cdot \lvert E_2 \rvert)$-time algorithm of~\cite{DBLP:conf/stringology/ShimohiraIBT11} for \LCSP\ is optimal under the same conditional lower bounds as for \SMLG~\cite{BI16,DBLP:conf/icalp/EquiGMT19,EquiMT21,DBLP:conf/sosa/GibneyHT21}, since the decision version of \SMLG\ (i.e.~whether there is an occurrence of the string) is linear-time reducible to \LCSP. In fact, the same holds also for \MSP. Nevertheless, there exists no analogous lower bound for \LRSP. Note that the three problems defined with occurrences as directed paths, i.e.\ visiting each vertex at most once, are NP-complete (\Cref{thm:np-path}), and hence in this paper we consider walk occurrences.

\SMLG, \LCSP\ and \LRSP\ have connections also to Automata Theory, since the spellings of all walks of a finite labeled graph form a regular language. Indeed, one can transform a labeled graph into an NFA by making every vertex a final state, adding a new initial state connected to all vertices, and moving the label of each vertex on each incoming edge: a string occurs in a labeled graph if it is accepted by its corresponding NFA; strings common to two labeled graphs correspond to strings accepted by the intersection of their corresponding NFAs; repeated strings of a labeled graph correspond to \emph{ambiguous words} of the resulting NFA, namely words having at least two accepting computations. 
\emph{Ambiguity} of automata (or its lack thereof) has been studied in the context of Descriptional Complexity Theory \cite{DBLP:journals/actaC/HanSS17,DBLP:conf/dcfs/Colcombet15,DBLP:journals/jucs/GoldstineKKLMW02,DBLP:journals/tc/BookEGO71}, not to be confused with Descriptive Complexity Theory. For example, the \emph{degree of ambiguity} of an NFA is the maximum number of accepting computations of any word by the automaton. While there are works about studying upper bounds of such metric~\cite{DBLP:journals/tcs/WeberS91,DBLP:journals/ijfcs/AllauzenMR11}, to the best of our knowledge there is no research on the \emph{longest} ambiguous words of an NFA.

As a first result on labeled graphs, we observe that on labeled directed trees (i.e.~rooted trees with all edges oriented away from the root) \LRSP\ and \LCSP\ can also be solved in linear time and space as an easy application of the tree counterparts of the suffix tree and the suffix array: the \emph{suffix tree of a tree}~\cite{DBLP:conf/focs/Kosaraju89a} and the \emph{XBW transform of a tree} (XBWT)~\cite{DBLP:conf/focs/FerraginaLMM05}.
The former, introduced by Kosaraju in 1989, generalizes the suffix tree to represent all suffixes of the strings spelled by the upwards paths of a given tree and admits linear-time construction algorithms~\cite{DBLP:journals/tcs/Breslauer98,DBLP:journals/ieiceta/Shibuya03}. The latter, introduced by Ferragina et al.\ in 2005, is an invertible transform, also computable in linear time, encoding a tree as an ordered list of elements each corresponding to a vertex: the order of these elements first follows the lexicographical ordering of the unique path from the parent of the corresponding vertex to the root of the tree, then the pre-order visit of the tree. \LRSP\ of a tree can be solved directly by either structure in the same way as \LRSP\ of a string, while \LCSP\ of two trees can be solved with a simple adaptation of either structure.

In this paper we introduce the \emph{labeled direct product} of two labeled graphs $G_1$ and $G_2$, denoted $G_1 \otimes G_2$, inspired from e.g.~the Cartesian product construction for the intersection language accepted by two finite state automata (\Cref{sec:labeleddirectproduct}). While not a completely novel idea, this product cleanly encodes each and every pair of walks of the input graphs spelling the same string and it appears as the right conceptual tool to \emph{optimally} solve string problems on graphs, in the same way as the suffix tree is a ubiquitous tool for optimal algorithms on text. Our results are as follows.

\subsection{Conceptually simpler and more efficient algorithms}
The current state-of-the-art algorithm for \SMLG\ was introduced in 1997 by Amir et al.~\cite{AmirLL00} in the context of hypertexts (i.e.\ directed graphs such that each vertex is labeled with a string). Given a $\Sigma$-labeled graph $G = (V,E,L)$ and a string $S = S[1]\cdots S[m] \in \Sigma^*$, the algorithm works by constructing a directed acyclic graph (DAG) $G'$ having vertex $v^i$ for each vertex $v \in V$ and for each position $i \in \lbrace 1, \dots, m \rbrace$, such that there is an edge between two vertices $v^i$, $w^{i+1}$ if $L(v) = S[i]$ and $(v,w) \in E$: \SMLG\ is then solved by finding and reporting a path of length $\lvert S \rvert - 1$ in $G'$. Instead, by treating the pattern $S$ as a labeled path $G_S$, we can solve \SMLG\ by simply finding a path of length $\lvert S \rvert - 1$ in $G \otimes G_S$ (\Cref{sec:labeleddirectproduct}). Since $G_S$ is a path, then $G \otimes G_S$ is a DAG, and thus such a path can be found in time linear in the size of $G \otimes G_S$. Our labeled product is a subgraph of the DAG $G'$ by Amir et al.: $G'$ considers mismatching vertices $v^i$ such that $L(v) \neq S[i]$ and it avoids computing the edges \emph{from} mismatching vertices but not the (remaining) edges \emph{to} mismatching vertices. Thus, their algorithm always takes time $\Omega( \lvert V \rvert \cdot \lvert S \rvert )$, even when $G \otimes G_S$ has smaller size, and takes time $\Theta( \lvert E \rvert \cdot \lvert S \rvert )$ for some families of inputs where $G \otimes G_S$ has smaller size.

Moreover, \LCSP\ on DAGs $G_1 = (V_1,E_1,L_1)$ and $G_2=(V_2,E_2,L_2)$ is equivalent to finding a path of maximum length of the DAG $G_1 \otimes G_2$, which is also solvable in time and space linear in the size of $G_1 \otimes G_2$ (\Cref{sec:labeleddirectproduct}). Thus, our \LCSP\ algorithm using $G_1 \otimes G_2$ is not only a conceptually simpler version of the $O(\lvert E_1 \rvert \cdot \lvert E_2 \rvert)$-time and $\Theta(\lvert V_1 \rvert \cdot \lvert V_2 \rvert)$-space dynamic programming algorithm of Shimohira et al.~\cite{DBLP:conf/stringology/ShimohiraIBT11} for \LCSP, but can also be faster and use less space, if $G_1 \otimes G_2$ has size $O( \lvert V_1 \rvert \cdot \lvert V_2 \rvert )$ or the alphabet $\Sigma$ has constant size (\Cref{remark:time}).
Otherwise, our algorithm implies a greater space usage, since it stores $G_1 \otimes G_2$: choosing not to store the edges of $G_1 \otimes G_2$ and instead computing them when needed results in a time and space complexity closer to that of the existing algorithm for \LCSP\ (\Cref{remark:space}).

\subsection{Simple solution to an open problem} In addition to providing simple algorithms on DAGs, the labeled product graph also allows for conceptually simple and efficient solutions on arbitrary graphs. For example, \LCSP\ on two graphs containing cycles was left open by Shimohira et al.~\cite{DBLP:conf/stringology/ShimohiraIBT11}, and in \Cref{sec:hardness} we show that it is solvable by just checking whether $G_1 \otimes G_2$ has a cycle, and if not, still finding a path of maximum length in $G_1 \otimes G_2$ (see \Cref{tab:complexities2} for a summary of the complexity results for \LCSP).
\begin{table}[tbp]
\begin{center}
\newcommand{\lcsppathpath}[0]{$\substack{%
	\\[.5mm]%
	\displaystyle O( \lvert V_1 \rvert + \lvert V_2 \rvert ) \\%
	\text{w/ suffix tree} \;
	\text{\cite{Gus97,DBLP:conf/focs/Farach97}}%
	\\[.5mm]%
}$}%
\newcommand{\lcsppathtree}[0]{$\substack{%
	\\[.5mm]%
	\displaystyle O( \lvert V_1 \rvert + \lvert V_2 \rvert ) \\%
	\text{w/ suffix tree~\cite{DBLP:conf/cpm/Akutsu93}, or} \\%
	\text{w/ XBWT~\cite{DBLP:conf/focs/FerraginaLMM05}}
	\\[.5mm]%
}$}%
\newcommand{\lcsptreetree}[0]{$\substack{%
	\\[.5mm]%
	\displaystyle O( \lvert V_1 \rvert + \lvert V_2 \rvert ) \\%
	\text{w/ suffix tree} \\ \text{of a tree~\cite{DBLP:journals/tcs/Breslauer98,DBLP:journals/ieiceta/Shibuya03}, or} \\%
	\text{w/ XBWT~\cite{DBLP:conf/focs/FerraginaLMM05}}%
	\\[.5mm]%
}$}%
\newcommand{\lcspdag}[0]{$\substack{%
	\\[.5mm]%
	\displaystyle O( \lvert E_1 \rvert \cdot \lvert E_2 \rvert ) \\%
	\text{w/ dynamic} \\ \text{programming} \\ \text{algorithm \cite{DBLP:conf/stringology/ShimohiraIBT11}, or} \\%
	\text{w/ labeled direct} \\ \text{product graph,} \;
	\text{\Cref{sec:labeleddirectproduct}}%
	\\[.5mm]%
}$}%
\newcommand{\lcspgraph}[0]{$\substack{%
	\\[.5mm]%
	\displaystyle O( \lvert E_1 \rvert \cdot \lvert E_2 \rvert ) \\%
	\text{w/ labeled direct} \\ \text{product graph} \\%
	\text{\Cref{sec:generalgraphs}}%
	\\[.5mm]%
}$}%
\newcommand{\lcspgraphpaths}[0]{$\substack{%
	\\[.5mm]%
	\displaystyle \text{NP-complete} \\%
	\text{\cite{DBLP:conf/icalp/EquiGMT19}, \Cref{sec:hardness}}%
	\\[.5mm]%
}$}%
	\caption{Summary, for some variants of \LCSP\ defined on walk occurrences on two graphs $G_1 = (V_1,E_1,L_1)$ and $G_2 = (V_2,E_2,L_2)$, of the time complexities. The linear-time algorithms assume an integer alphabet and the quadratic-time algorithms are optimal under \OVH\ (\Cref{thm:LRSPcomplexity}).}\label{tab:complexities2}%
\begin{tabular}{r|cccc}
	\toprule
	$G_1 \backslash G_2$ & path & tree & DAG & graph \\\midrule
	path & \lcsppathpath & \lcsppathtree & \multicolumn{2}{c}{\multirow{4}*{\lcspdag}} \\\cmidrule{1-3}
	tree & -- & \lcsptreetree \\\cmidrule{1-3}
	DAG & -- & -- \\\cmidrule{1-5}
	graph & -- & -- & -- & \lcspgraph \\
	\bottomrule
\end{tabular}
\end{center}
\end{table}

\subsection{Solutions to new problems} The labeled direct product also allows for solutions to related problems. For \MSP\ on DAGs we analogously find paths of maximum length from some vertices of $G_1 \otimes G_2$ (\Cref{thm:alg-dags}). We generalize this algorithm on arbitrary graphs by computing the strongly connected components (SCCs) of $G_1 \otimes G_2$ and by checking a condition analogous to that of \LCSP\ for every vertex $v$ of $G_1$ (\Cref{thm:main-algorithm2}). These algorithms use time and space linear in the size of $G_1 \otimes G_2$.

\LRSP\ on a DAG $G$ is equivalent to finding paths of maximum length passing through specific vertices of $G \otimes G$~(\Cref{thm:alg-dags}). On arbitrary graphs, we use further interesting connections between purely graph-theoretic concepts of the labeled product graph (SCCs) and string-theoretic ones (non-deterministic vertices). The difference of \LRSP\ with respect to \LCSP\ and \MSP\ is that the problem may admit repeated strings of infinite length or repeated strings of unbounded lengths---these two scenarios may not coincide. Even if the difference between these two concepts may seem artificial, their study is necessary for the natural characterization of \LRSP\ solutions. Indeed, in \Cref{sec:generalgraphs} we show that such cases can be efficiently identified (infinite repeated strings can be identified in $G \otimes G$ by checking reachability from a certain set of vertices to a non-trivial SCC, while repeated strings of unbounded length can be identified by checking reachability from a non-trivial SCC to some non-deterministic vertex). If none of these cases happen, we show that the problem is solvable with the DAG algorithm for \LRSP\ on an acyclic subgraph of $G \otimes G$~(\Cref{thm:main-algorithm}). The entire procedures take time and space linear in the size of $G \otimes G$. In addition, we can also output a \emph{linear-size} representation of a longest repeated string, infinite or not.

\subsection{Optimality under conditional lower bounds} In \Cref{sec:hardness} we show that the above algorithms of worst-case quadratic-time complexity are also conditionally optimal. First, we note how the quadratic lower bounds of \cite{DBLP:conf/icalp/EquiGMT19,DBLP:conf/sosa/GibneyHT21} imply the same quadratic lower bounds for \LCSP\ and \MSP.
Second, in \Cref{thm:SMLG-reduction} of \Cref{sec:hardness}, we show that on DAGs that are \emph{deterministic} (i.e.~the labels of the out-neighbors of every vertex are all distinct) the \SMLG\ problem has a linear-time reduction to \LRSP, which thus implies the same lower bounds for \LRSP\ as in~\cite{DBLP:conf/icalp/EquiGMT19,DBLP:conf/sosa/GibneyHT21} (holding also for deterministic DAGs). To the best of our knowledge such a reduction does not exist when the problems are defined on strings. Third, in \Cref{thm:LRSPcomplexity} we show that, under the Orthogonal Vectors Hypothesis (\OVH)~\cite{Williams05}, there can be no truly sub-quadratic algorithm solving \LRSP, even when the graph is a DAG, with vertex labels from a binary alphabet, maximum in-degree and out-degree of any vertex at most 2, and is deterministic. Our reduction for \LRSP\ is simpler than that of~\cite{DBLP:conf/icalp/EquiGMT19}, but with an interesting difficulty arising from the fact that we must encode the orthogonal vectors input in the \emph{same} graph, and must ensure that the occurrences of the longest repeated string are distinct.

Moreover, in the same way as the labeled direct product graph is a general tool for obtaining algorithms, the construction behind our reduction could also be a general approach to obtain conditional lower bounds for string problems on graphs. For example, our \OVH\ construction (simpler than \cite{DBLP:conf/icalp/EquiGMT19}) also provides a conditional lower bound for \LCSP, and in \Cref{cor:MSP*} we show our reduction also proves the same conditional lower bound for a variant of \MSP.

\subsection{The full complexity picture of \LRSP} Finally, since on directed graphs \LRSP\ turned out the most complex problem to solve, in \Cref{sec:undirected} we complete its complexity picture by studying it also on undirected graphs, by similarly considering undirected paths, trees and graphs, and the path and walk variants of the problem (see \Cref{tab:complexities}). While these results are simpler than for directed graphs, they exhibit some interesting complexity dichotomies on analogous classes of graphs. For example, for walk occurrences the problem is \emph{linear-time} solvable on general undirected graphs, as opposed to having a conditional quadratic-time lower bound on general \emph{directed} graphs). Note that the \SMLG\ problem has the same complexity on both directed and undirected graphs~\cite{DBLP:conf/icalp/EquiGMT19}, making this dichotomy for \LRSP\ more interesting. Moreover, when defined on paths, we obtain only a quadratic-time algorithm on undirected trees, even though \LRSP\ is linear on \emph{directed} trees. As such, we put forward as an interesting open problem either improving this complexity, or proving a lower bound.
\begin{table}[tbp]
	\caption{Summary, for all the variants of \LRSP\ on a graph $(V,E,L)$, of the time complexities. The linear-time algorithms assume an integer alphabet and the quadratic-time algorithms for directed graphs are optimal under \OVH\ (\Cref{thm:LRSPcomplexity}). We leave as an open problem improving our solution to \LRSP\ on undirected trees, when the string occurrences are paths, or proving it is conditionally tight.}\label{tab:complexities}
	\centering
\newcommand{\directedlineargraphs}[0]{$\substack{%
	\\[.5mm]%
	\displaystyle O(\lvert V \rvert) \\%
	\text{w/ suffix tree} \;
	\text{\cite{Gus97,DBLP:conf/focs/Farach97}}%
	\\[.5mm]%
}$}%
\newcommand{\directedtrees}[0]{$\substack{%
	\\[.5mm]%
	\displaystyle O(\lvert V \rvert) \\%
	\text{w/ suffix tree of a tree~\cite{DBLP:journals/tcs/Breslauer98,DBLP:journals/ieiceta/Shibuya03}, or} \\
	\text{w/ XBW transform of a tree~\cite{DBLP:conf/focs/FerraginaLMM05}} \;
	\\[.5mm]%
}$}%
\newcommand{\directeddags}[0]{$\substack{%
	\\[.5mm]%
	\displaystyle O \big( \lvert E \rvert^2 \big) \\%
	\text{w/ labeled direct} \\ \text{\vphantom{/}self-product graph} \\%
	\text{\vphantom{/}\Cref{sec:alg-dags}}%
	\\[.5mm]%
}$}%
\newcommand{\directedonwalks}[0]{$\substack{%
	\\[.5mm]%
	\displaystyle O \big( \lvert E \rvert^2 \big) \\%
	\text{w/ labeled} \\ \text{direct}\\
	\text{self-product} \\ \text{graph\vphantom{/}} \\%
	\text{\Cref{sec:generalgraphs}}%
	\\[.5mm]
}$}%
\newcommand{\directedonpaths}[0]{$\substack{%
	\\[.5mm]%
	\displaystyle \text{NP-complete} \\%
	\text{\Cref{thm:np-path}}%
	\\[.5mm]
}$}%
\newcommand{\undirectedlineargraphsonpaths}[0]{$\substack{%
	\\[.5mm]%
	\displaystyle O(\lvert V \rvert) \\%
	\text{w/ suffix tree} \\%
	\text{\Cref{sec:undirected}}%
	\\[.5mm]
}$}%
\newcommand{\undirectedtreesonwalks}[0]{$\substack{%
	\\[.5mm]%
	\displaystyle O(\lvert V \rvert) \\%
	\text{w/ repeated strings} \\%
	\text{of length 2 check\vphantom{/}} \\%
	\text{\Cref{sec:undirected}}%
	\\[.5mm]
}$}%
\newcommand{\undirectedonwalks}[0]{$\substack{%
	\\[.5mm]%
	\displaystyle O(\lvert E \rvert) \\%
	\text{w/ repeated strings} \\%
	\text{of length 2 check\vphantom{/}} \\%
	\text{\Cref{sec:undirected}}%
	\\[.5mm]%
}$}%
\newcommand{\undirectedtreesonpaths}[0]{$\substack{%
	\\[.5mm]%
	\displaystyle O \big( \lvert V \rvert^2 \big) \\%
	\text{w/ reduction to \LRSP} \\%
	\text{on directed trees}\vphantom{/} \\%
	\text{\Cref{sec:undirected}}
	\\[.5mm]%
}$}%
\newcommand{\undirectedonpaths}[0]{$\substack{%
	\displaystyle \text{NP-complete} \\%
	\text{\Cref{sec:undirected}}
}$}%
	\begin{tabular}{c|c c c c}
		\toprule
		\multirow{2}{*}{Graph Class} & \multicolumn{4}{c}{Graph Type} \\
		& \multicolumn{2}{c}{directed} & \multicolumn{2}{c}{undirected} \\
		\midrule
		paths & \multicolumn{2}{c|}{\directedlineargraphs} &                        & \undirectedlineargraphsonpaths \\ \cmidrule{1-3}\cmidrule{5-5}
		trees         & \multicolumn{2}{c|}{\directedtrees}        & \multirow{-4}{*}{\undirectedtreesonwalks} & \undirectedtreesonpaths \\ \cmidrule{1-5}
		DAGs          & \multicolumn{2}{c|}{\directeddags} & -- & -- \\ \cmidrule{1-5}
		graphs        & \directedonwalks & \multicolumn{1}{c|}{\directedonpaths} & \undirectedonwalks & \undirectedonpaths \\
		\midrule
		& on walks & \multicolumn{1}{c}{on paths} & on walks & on paths \\
		& \multicolumn{4}{c}{Occurrence definition} \\
		\bottomrule
	\end{tabular}
\end{table}

\subsection{Notation and preliminaries}

Given a non-empty and finite alphabet $\Sigma$, we denote with $\Sigma^*$ and $\Sigma^\omega$ the set of all finite and infinite strings over $\Sigma$, respectively. For convenience, we also define $\Sigma^+ \coloneqq \Sigma \setminus \lbrace \varepsilon \rbrace$, with $\varepsilon$ the empty string. We say that $\Sigma$ is an \emph{integer alphabet} if it contains integers from a range that is linear-sized with respect to the input of the problem at hand, allowing linear-time lexicographical sorting. Given the $\Sigma$-labeled graph $G = (V,E,L)$, a \emph{walk} in $G$ is any finite or infinite sequence of vertices $p = (p_0, p_1, p_2, \dots)$, such that there is an edge from any $p_i$ to its successor in $p$. If all vertices of $p$ are pairwise distinct, then $p$ is called a \emph{path}. The \emph{length} of a finite walk $p$ is its number of edges. Just as strings can be concatenated to form longer strings, walks can be concatenated to form longer walks under the condition that the result is still a walk in $G$. Two walks $p = (p_0,p_1,p_2,\dots)$, $q = (q_0,q_1,q_2,\dots)$ in $G$ are \emph{distinct}, in symbols $p \neq q$, if there is an index $i$ such that $p_i \neq q_i$.  A finite (resp.\ infinite) \emph{string occurring in} $G$ (or simply, a string of $G$) is any string $S \in \Sigma^*$ (resp.\ $S \in \Sigma^\omega$) such that there is a finite (resp.\ infinite) walk $p = (p_0, p_1, p_2, \dots)$ in $G$ with $S = L(p) \coloneqq L(p_0)L(p_1)L(p_2)\dots$. We say that $p$ is an \emph{occurrence} of $S$ in $G$, that $p$ \emph{spells} $S$ in $G$ or that $S$ has a \emph{match} in $G$. A string $S$ occurring in $G$ is \emph{repeated} if there are at least two distinct occurrences of $S$ in $G$, in symbols $\exists p,q \; \text{walks in $G$}$ such that $p \neq q \,\wedge\, L(p) = L(q) = S$. Throughout the paper, we will assume that every vertex has at least one in-neighbor or out-neighbor, so it holds that $\lvert V \rvert \le 2 \lvert E \rvert$, $\lvert V \rvert \in O( \lvert E \rvert )$ and we can simplify a complexity bound such as $O(\lvert V \rvert + \lvert E \rvert)$ into $O(\lvert E \rvert)$.

\begin{remark}\label{remark}
In solving \LCSP, \MSP\ and \LRSP\ we can assume that for any input labeled graph $G = (V,E,L)$ it holds that $\lvert V \lvert \in O( \lvert E \lvert)$, because:
\begin{itemize}
	\item the problems become trivial when considering only paths of length $0$, in the sense that there is a common or repeated string of length one if and only if there are different vertices labeled with the same character in the respective graphs; if in $G$ there are vertices without both incoming and outgoing edges, they can be treated separately since the strings they generate have all length $1$, thus we will assume throughout the rest of the paper that every vertex $v \in V$ has at least one incoming or outgoing edge, meaning that
	$\lvert E \rvert \ge \lvert V \rvert / {2}$.

	\item the answer to \LCSP\ and \LRSP\ is the empty string $\varepsilon$ if and only if the sets of labels used in $G_1$ and $G_2$ do not intersect, or if each vertex of $G$ is labeled with a different character (implying $\lvert \Sigma \rvert \ge \lvert V \rvert$); this can be easily checked assuming we are working with an integer alphabet, if not it is still $O( \lvert V \rvert \log \lvert V \rvert)$, so we will assume that there is a common or repeated string of length 1, unless stated otherwise.
\end{itemize}
\end{remark}

\section{The Labeled Direct Product}
\label{sec:labeleddirectproduct}

	Recall that the direct product of two graphs is the graph whose vertex set is the Cartesian product of the vertex sets of the initial graphs where we have an edge between two vertices if there are corresponding edges in the initial graphs between vertices on the first component and between vertices on the second component. 
	This product has been studied in the literature in both the undirected and directed setting, under the names \emph{conjunction}, \emph{tensor product}, \emph{Kronecker product}, and others (see \cite[p.~21]{DBLP:books/daglib/0070576} and \cite{harary1967boolean}). We will use instead the \emph{labeled} direct product of $G_1$ and $G_2$, obtained as the subgraph of the direct product of $G_1$ and $G_2$ induced by the vertices for which their two components have the same label.
	Although this notion is similar to the automaton recognizing the intersection of two automata (see \cite{DBLP:journals/ibmrd/RabinS59}), the key difference is that the labeled direct product graph does not contain any pair of edges/transitions with mismatching labels.

\subsection{Definition and basic properties}
 Consider the following definition, and see also \Cref{fig:labeled-product,fig:labeled-self-product}.

\begin{definition}[Labeled direct product graph]
    \label{def:labeled-direct-self-product}
	Given two $\Sigma$-labeled graphs $G_1 = (V_1,E_1,L_1)$ and $G_2 = (V_2,E_2,L_2)$, we define \emph{the labeled direct product graph $G_1 \otimes G_2 = (V', E', L')$}, where:
	\begin{gather*}
		V' = \big\lbrace (u,v) \in V_1 \times V_2 : L_1(u) = L_2(v) \big\rbrace,\\
		E' = \Big\lbrace
			\big( (u,v), (u',v') \big)  \in V' \times V' \,:\,
			(u,u') \in E_1 \wedge (v,v') \in E_2
		\Big\rbrace,
	\end{gather*}
	 and $L'$ is defined so that $L'(u,v) = L_1(u) = L_2(v)$ for each $(u,v) \in V'$. 
\end{definition}

\begin{figure}[tbp]
	\centering
    \newbox\mybox
    \setbox\mybox=\hbox{\includegraphics[scale=1]{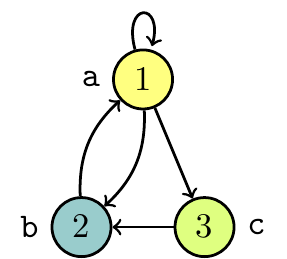}}
	\begin{minipage}{\wd\mybox}\includegraphics[scale=1]{Figures/g1example}\end{minipage}
	\quad
	\setbox\mybox=\hbox{\includegraphics[scale=1]{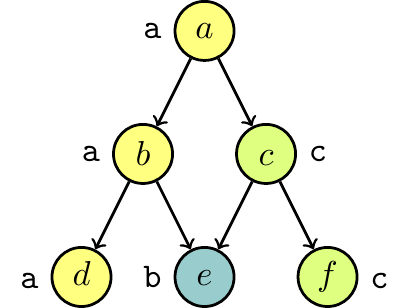}}
	\begin{minipage}{\wd\mybox}\includegraphics[scale=1]{Figures/g2example}\end{minipage}
	\quad
	\setbox\mybox=\hbox{\includegraphics[scale=1]{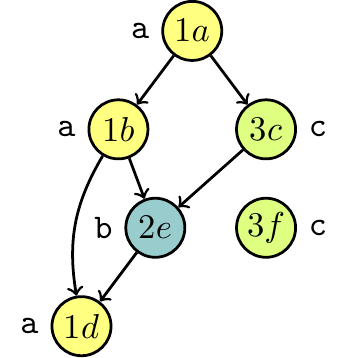}}
	\begin{minipage}{\wd\mybox}\includegraphics[scale=1]{Figures/g1g2product}\end{minipage}
	\caption{An example of two $\lbrace \mathtt{a},\mathtt{b},\mathtt{c} \rbrace$-labeled graphs $G_1$, $G_2$ and their labeled direct product graph $G_1 \otimes G_2$ on the right. Since $G_2$ is a DAG, $G_1 \otimes G_2$ is a DAG as well.}\label{fig:labeled-product}
\end{figure}
\begin{figure}[tbp]
	\centering
    \setbox\mybox=\hbox{\includegraphics[scale=1]{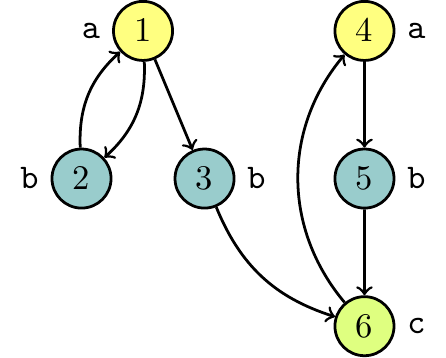}}
	\begin{minipage}{\wd\mybox}\includegraphics[scale=1]{Figures/graphnonproductexample}\end{minipage}
	\qquad
	\setbox\mybox=\hbox{\includegraphics[scale=1]{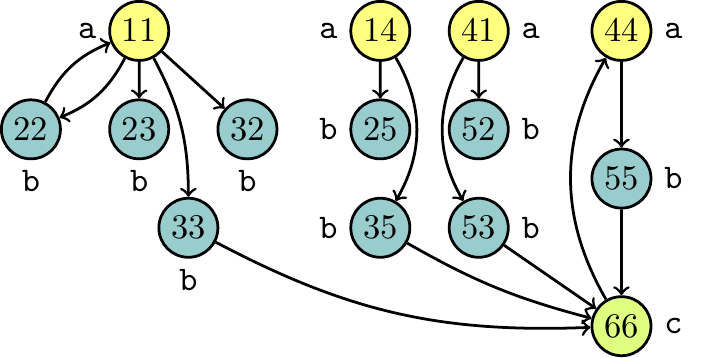}}
	\begin{minipage}{\wd\mybox}\includegraphics[scale=1]{Figures/graphproductexample}\end{minipage}
	\caption{A $\lbrace \mathtt{a}, \mathtt{b}, \mathtt{c} \rbrace$-labeled directed graph $G$ (left) and its labeled direct self-product $G \otimes G$ (right). Since $G$ is cyclic, $G \otimes G$ is cyclic as well.}\label{fig:labeled-self-product}
\end{figure}

Given a vertex $q = (u, v) \in V'$, let $\pi_1 (q) \coloneqq u$ and $\pi_2 (q) \coloneqq v$. Given a walk $q = \big((p_0, p_0'), (p_1, p_1'), \dots \big)$ in $G_1 \otimes G_2$, we denote with $\pi_1(q)$ and $\pi_2(q)$ the walks $(p_0, p_1, \dots)$ and $(p_0', p_1', \dots)$ in $G_1$ and $G_2$, respectively. We state the following basic fact about the correspondence between the pairs of walks in $G_1$ and $G_2$ and the walks in $G_1 \otimes G_2$. In particular, this implies that the projections of any cycle in $G_1 \otimes G_2$ are two cycles in $G_1$ and $G_2$ reading the same string and vice versa.

\begin{remark}\label{lem:graphproductproperty}
	Given $G_1$, $G_2$ $\Sigma$-labeled graphs, for each pair $(p_0, p_1, \dots)$, $(p_0', p_1', \dots)$ of finite (resp.\ infinite) walks in $G_1$ and $G_2$, respectively, reading the same finite (resp.\ infinite) string $S \in \Sigma^*$ (resp.\ $S \in \Sigma^\omega$), $p \otimes p' \coloneqq \big( (p_0, p_0'), (p_1, p_1'), \dots \big)$ is a finite (resp.\ infinite) walk in $G_1 \otimes G_2$ reading $S$ and vice versa.
\end{remark}

Since all the algorithms we develop consist in analyzing the labeled direct product of the input graphs, we must take great care in the time and space spent on its construction. Moreover, in \Cref{remark:space} we show that its size can be computed efficiently, making the run time of our algorithms predictable.
\begin{remark}\label{remark:time}
The construction of $G_1 \otimes G_2 = (V',E',L')$ takes $O(\lvert V_1 \rvert \cdot \lvert V_2 \rvert + \lvert E_1 \rvert \cdot \lvert E_2 \rvert)$ time and space, because each pair of vertices and each pair of edges need to be considered at most once. Assuming $\Sigma$ to be an integer or a constant-size alphabet we can do better than the naive construction algorithm with respect to time or space:
\begin{itemize}
    \item if $\Sigma$ is an integer alphabet, by first sorting lexicographically the lists of edges of $G_1$ and $G_2$, the product $G_1 \otimes G_2$ can then be built in linear-time with respect to its size, by simply pairing all edges of $G_1$ and $G_2$ with matching labels;
    \item if $\Sigma$ has constant size, there is no need to store the edges of $G_1 \otimes G_2$, since for all $a \in \Sigma$ we can report in time linear in the solution all $a$-labeled out-neighbors of any vertex $(u,v)$ by pairing all $a$-labeled out-neighbors of $u$ and $v$ in $G_1$ and $G_2$, respectively;
    \item if $\Sigma$ is an integer alphabet, preprocessing $G_1$, $G_2$ in order to report the (number of) out-neighbors of any vertex $(u,v)$ of $G_1 \otimes G_2$ is equivalent to the \textsf{SetIntersection} problem, for which Goldstein et al.\ proved conditional lower bounds on the trade-off between the space and time used in its solution~\cite{DBLP:conf/wads/GoldsteinKLP17}; if we choose not to store at all the edges of $G_1 \otimes G_2$, the algorithms exploiting $G_1 \otimes G_2$ will then take $\Theta(\lvert V' \rvert)$ space and $O(\lvert V' \rvert + \lvert E_1 \rvert \cdot \lvert E_2 \rvert)$ time.\footnote{We speculate that a careful implementation of our algorithms using bitvectors might take $\Theta( \lvert V' \rvert )$ space and $O\big(\lvert V' \rvert + \lvert E' \rvert + \lvert V' \rvert \cdot \lvert \Sigma \rvert / \log (\lvert \Sigma \rvert)\big)$ time. The conditional lower bounds by Goldstein et al.\ ignore logarithmic factors, so this would not be a contradiction.}
\end{itemize}
\end{remark}

\begin{remark}
Since vertices and edges in $G_1 \otimes G_2$ correspond to vertices and edges in $G_1$ and $G_2$ \emph{with matching labels}, the size of $G_1 \otimes G_2$ could be much less than $\lvert V_1 \rvert \cdot \lvert V_2 \rvert + \lvert E_1 \rvert \cdot \lvert E_2 \rvert$ in practice, or for some families of labeled graphs. In particular, if each $a \in \Sigma$ is the label of at most $O(1)$ pairs of vertices in $V_1 \times V_2$ then $G_1 \otimes G_2$ has size $O( \lvert V_1 \rvert + \lvert V_2 \rvert + \lvert E_1 \rvert + \lvert E_2 \rvert )$: this is not in contradiction with the conditional lower bounds of \Cref{sec:hardness} because the graph obtained in the reduction of \Cref{thm:LRSPcomplexity} uses only two labels, $\Theta(\lvert V \rvert)$ times each.
\end{remark}
\begin{remark}\label{remark:space}
If $\Sigma$ is an integer alphabet, the size of $G_1 \otimes G_2$ can be computed in time linear in the size of the input graphs $G_1$ and $G_2$. Indeed, let $V_1^a$, $V_2^a$ be the sets of $a$-labeled vertices of $G_1$, $G_2$, respectively, and let $E_i^{a,b}$ be the set edges of $G_i$ connecting an $a$-labeled vertex to a $b$-labeled vertex, with $i = 1,2$. Then it is easy to see that
\begin{gather*}
    \lvert V' \rvert = \prod_{a \in \Sigma} \lvert V_1^a \rvert \cdot \lvert V_2^a \rvert
    \quad\text{and}\quad
    \lvert E' \rvert = \prod_{a,b \in \Sigma} \lvert E_1^{a,b} \rvert \cdot \lvert E_2^{a,b} \rvert
\end{gather*}
and that $\lvert V' \rvert + \lvert E' \rvert$ can be easily computed after sorting the vertex and edge sets of $G_1$ and $G_2$. Note that if $\Sigma$ has constant size, the size of $G_1 \otimes G_2$ can be found in constant time after the independent sorting of $G_1$ and $G_2$.
\end{remark}

\subsection{Optimal algorithms for DAGs}
\label{sec:alg-dags}

We first consider the case when the direct product graph is a DAG. Note that $G \otimes G$ is DAG if and only if $G$ itself is, $G_1 \otimes G_2$ is a DAG if at least one between $G_1$ and $G_2$ is a DAG, but $G_1 \otimes G_2$ might be a DAG even if both $G_1$ and $G_2$ contain cycles.

Thanks to \Cref{lem:graphproductproperty}, \SMLG, \LCSP, \MSP\ and \LRSP\ can be solved by finding paths of maximum length in the corresponding direct product graph: an occurrence of pattern $S$ in graph $G$ corresponds to a path of length $\lvert S \rvert - 1$ in $G \otimes G_S$, where $G_S$ is a labeled path of $\lvert S \rvert$ vertices spelling $S$; a longest common string of $G_1$ and $G_2$ is spelled by a path of maximum length in $G_1 \otimes G_2$; the matching statistics $\MS(v)$ of $G_1$ and $G_2$, with $v \in V_1$, is equal to one plus the length of a path of maximum length in $G_1 \otimes G_2$ starting from any vertex in $\lbrace v \rbrace \times V_2$; a longest repeated string of a DAG $G$ is spelled by a path of maximum length of $G \otimes G$ visiting at least one vertex $(u,v)$ such that $u \neq v$.

Indeed, for every vertex $(u,v)$ of the product graph $(V',E',L')$ we can compute by dynamic programming the length $\ell^+(u,v)$ of the longest path starting at $(u,v)$:
\begin{itemize}
    \item $\SMLG$ is solved by finding a path of length $\lvert S \rvert$ starting from a vertex $(u,v)$ such that $\ell^+(u,v) = \lvert S \rvert$;
    \item $\LCSP$ is solved by finding a vertex $(u,v)$ of $G_1 \otimes G_2$ such that $\ell^+(u,v)$ has maximum value and by retrieving the string corresponding to a path of length $\ell^+(u,v)$ starting at $(u,v)$;
    \item $\MSP$ is solved by finding for each $v \in V_1$ the maximum value of $\ell^+(v,w) + 1$, with $(v,w)$ a vertex of $G_1 \otimes G_2$, and this can be done by iterating once over all vertices in $V'$.
\end{itemize}
We can analogously compute for each $(u,v) \in V'$ the length $\ell^-(u,v)$ of the longest path in $(V',E',L')$ ending at $(u,v)$:
\begin{itemize}
    \item $\LRSP$ is solved by iterating over all vertices $(u,v)$ of $G \otimes G$ such that $u \neq v$, and obtaining the length of the longest repeated string in $G$ whose occurrences pass through the distinct vertices $u$ and $v$, as $\ell^+(u,v) + \ell^-(u,v) +1$ (a longest repeated string of this length can then be retrieved).
\end{itemize}

\begin{theorem}
\label{thm:alg-dags}
    Given $G_1 = (V_1,E_1,L_1)$, $G_2 = (V_2,E_2,L_2)$ $\Sigma$-labeled directed graphs, \LCSP\ and \MSP\ on $G_1$, $G_2$ are solvable in $O(\lvert E_1 \rvert \cdot \lvert E_2 \rvert)$ time and taking $O(\lvert V_1 \rvert \cdot \lvert V_2 \rvert)$ words in space. Analogously, $\LRSP$ on a $\Sigma$-labeled graph $G = (V,E,L)$ is solvable in $O(\lvert E \rvert^2)$ time and taking $O(\lvert V \rvert^2)$ words in space. For all three problems plus \SMLG, if the product graph is given, then the solution takes linear time in the size of the product graph.
\end{theorem}

\section{Optimal Algorithms for General Graphs}
\label{sec:generalgraphs}

Since \LCSP, \MSP\ and \LRSP\ defined on paths are NP-complete (see \Cref{sec:hardness}), we focus here on the three problems defined on walks. If we deal with graphs containing cycles, then the length of the walks and strings to consider is not bounded anymore so we modify the three problems to require the detection of the relative cases. In fact, we will show that in all cases we can also report a \emph{linear-size} representation of the corresponding common or repeated strings. As we stated in the introduction, the three problems admit worst-case quadratic-time solutions based on the labeled direct product graph, i.e.\ $G_1 \otimes G_2$ for \LCSP\ and \MSP\ and $G \otimes G$ for \LRSP.
\begin{definition}
    Given a labeled direct product graph $G_1 \otimes G_2$ or $G \otimes G$, we define: \begin{itemize}
        \item $V'_{\text{\emph{cyc}}}$ as the set of all vertices of the product graph involved in a cycle, namely those belonging to an SCC consisting of at least two vertices;
    \end{itemize}
    also, for $G \otimes G = (V',E',L')$ we define:
    \begin{itemize}
        \item $V'_{\text{\emph{diff}}}$ as the set of vertices $(u,v) \in V'$ with $u \neq v$;
        \item $V'_{\text{\emph{ndet}}}$ as the set of all vertices $(v,v) \in V'$ with $v$ a \emph{non-deterministic} vertex of $G$, that is, $v$ has two out-neighbors labeled with the same character.
    \end{itemize}
\end{definition}

\begin{figure}[tbp]
	\centering
	\begin{minipage}{.45\textwidth}\centering\includegraphics[scale=1]{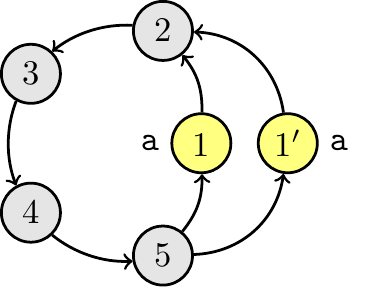}\end{minipage}%
	\begin{minipage}{.45\textwidth}\centering\includegraphics[scale=1]{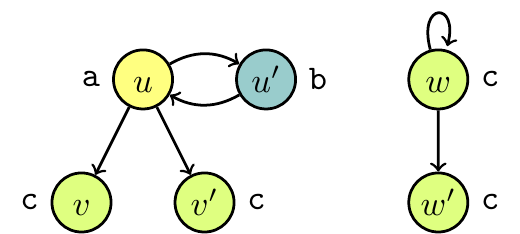}\end{minipage}
	\caption{An example (left) of a non-deterministic graph $G$ with two distinct cycles $(1,2,3,4,5)$ and $(1',2,3,4,5)$; their (uncountably many) infinite repetition generates the same infinite string; and a non-deterministic graph (right) with no infinite repeated strings but with  finite repeated strings of unbounded length, precisely of the form $(\mathtt{ba})^k\mathtt{c}$, $(\mathtt{ab})^k\mathtt{ac}$ and $\mathtt{c}^k\mathtt{c}$, for every $k \ge 0$.}\label{fig:unboundedrep}
\end{figure}

\subsection{\LCSP\ and \MSP\label{sec:LCSP-MSP}}
Since the graphs can contain cycles, the common strings in \LCSP\ and \MSP\ can now have infinite length. 
The algorithm solving \LCSP\ on any two $\Sigma$-labeled graphs $G_1$, $G_2$ consists of the following simple checks in $G_1 \otimes G_2$:
\begin{description}
	\item[\textbf{Infinite length common strings}] Check if $G_1 \otimes G_2$ contains a cycle; if so, return $(i)$ the string spelled by any cycle and $(ii)$ the symbol $\omega$; otherwise
	\item[\textbf{Finite length common strings}] Proceed as in the algorithm for the DAG case from \Cref{sec:alg-dags} on $G_1 \otimes G_2$.
\end{description}
The correctness of this algorithm follows from the fact that there is a common string of infinite length if and only if there is a common string of infinite length of the form $S^\omega$ (see also \Cref{lem:periodicstrings} below).

\MSP\ can be solved as well by studying the Strongly Connected Components (SCCs) of $G_1 \otimes G_2 = (V', E', L')$:
\begin{description}
	\item[\textbf{Infinite length matching statistics}] For all $(u,v) \in V'_{\text{cyc}}$ set $\MS(u) = \infty$.
	\item[\textbf{Finite length matching statistics}] Proceed for the remaining vertices of $V_1$, i.e.\ the vertices $u \in V_1$ such that no $(u,v) \in V'$ is also in $V'_{\text{cyc}}$, as in the algorithm for the DAG case from \Cref{sec:alg-dags}. Note that in this second step we consider an acyclic subgraph of $G_1 \otimes G_2$.
\end{description}

\begin{theorem}
\label{thm:main-algorithm2}
    The above algorithms correctly solve \LCSP\ and \MSP. Moreover, if $G_1 \otimes G_2$ is given, they can be implemented to run in time linear in the size of $G_1 \otimes G_2$. If not, they run in time $O( \lvert E_1 \rvert \cdot \lvert E_2 \rvert )$, where $E_1$ and $E_2$ are the edge sets of $G_1$ and $G_2$, respectively.
\end{theorem}

\subsection{\LRSP\label{sec:LRSP}}
In \LRSP\ on general graphs we have one of the following \emph{three} cases, as seen in \Cref{fig:unboundedrep}:
\begin{enumerate}
	\item The graph has an infinite repeated string.
	\item The graph does not have any infinite repeated string, but the length of the repeated strings is unbounded.
	\item The length of the longest repeated string is bounded and there are repeated strings of a finite maximum length (as is the case for texts, trees and DAGs).
\end{enumerate}

An undesirable feature of infinite strings is that they can be aperiodic. However, analogously to some results of Büchi automata theory stating that the ``important'' strings are \emph{ultimately periodic}~\cite[p.~137] {DBLP:books/el/leeuwen90/Thomas90}, that is they are of the form $RS^\omega$ with $R \in \Sigma^*$ and $S \in \Sigma^+$,
in \Cref{lem:periodicstrings} we show that the presence of infinite repeated strings can be detected by looking for ultimately periodic strings. Its easy proof, which we omit, finds a cycle in $G \otimes G$ used by two distinct occurrences of an infinite repeated string $w$ (since $G$ is finite): we can build $RS^\omega$ by identifying $R$ as the prefix of $w$ spelled by the path reaching the cycle and $S$ as the string spelled by the cycle.

\begin{lemma}
\label{lem:periodicstrings}
	Given a $\Sigma$-labeled graph $G = (V,E,L)$, there is an infinite repeated string occurring in $G$ if and only if there is a string $RS^\omega \in \Sigma^\omega$ in $G$ spelled by two distinct walks $rs^\omega$ and $r's'^\omega$ in $G$, with $R = L(r) = L(r') \in \Sigma^*$ and $S = L(s) = L(s') \in \Sigma^+$.
\end{lemma}

The two distinct walks spelling $RS^\omega$ provided by \Cref{lem:periodicstrings} imply the existence of a walk in $G \otimes G$ passing trough a vertex $q$ in $V'_{\text{diff}}$ and reaching a vertex $q'$ in $V'_{\text{cyc}}$. Note that $q = q'$ can hold, in which case the infinite repeated string is of the form $S^\omega$. Thus, we obtain:

\begin{corollary}[Infinite repeated strings]\label{cor:infinite}
	$G$ has an infinite repeated string if and only if any $q \in V'_{\text{diff}}$ reaches any $q' \in V'_{\text{cyc}}$, and if $R$ is the spelling of a path from $q$ to $q'$ and $S$ is the spelling of a cycle starting and ending in $q'$, then $RS^\omega$ is an infinite repeated string in $G$.
\end{corollary}

If the graph has no infinite repeated string, the remaining difficulty is that of repeated strings of unbounded length. Formally, we say that $G$ has \emph{repeated strings of unbounded length} if for each $n \in \mathbb{N}$ there is a repeated string $S \in \Sigma^*$ occurring in $G$ such that $\lvert S \rvert > n$. It is easy to see that in the graph of \Cref{fig:unboundedrep} (right) there are no infinite repeated strings and the unbounded repeated strings are of the form $R^+ S$, with $R,S \in \Sigma^+$. Indeed, these unbounded strings are of this form and their occurrences have a common prefix after which they diverge. This divergence happens by visiting a non-deterministic vertex, as shown by the next two results.
	
\begin{lemma}\label{lem:unboundedfamilies}
	Given a $\Sigma$-labeled graph $G = (V,E,L)$ without infinite repeated strings, $G$ has repeated strings of unbounded length if and only if there are $R,S \in \Sigma^+$ such that $R^m S$ is repeated in $G$ for each $m \ge 1$.
\end{lemma}
\begin{proof}
	($\Leftarrow$) This side is trivial. ($\Rightarrow$) Let $G \otimes G = (V', E', L')$. If $G$ has repeated strings of unbounded length, then there must be some $k \in \mathbb{N}$ such that there is a repeated string $T \in \Sigma^*$ of length $k > \lvert V' \rvert$, with $p = (p_0, p_1, \dots, p_{k-1})$ and $p' = (p_0', p_1', \dots, p_{k-1}')$ two distinct occurrences of $T$ in $G$. Then $q \coloneqq (q_0, q_1, \dots, q_{k-1}) \coloneqq p \otimes p'$ is a walk in $G \otimes G$ visiting more than $\lvert V' \rvert$ vertices, so by the pigeonhole principle there must be a vertex visited more than once: let $j,j' \in \mathbb{N}$ be two indices such that
	$0 \le j < j' \le k-1$
	and
	$q_j = q_{j'}$.
	Since $p$ and $p'$ are distinct walks in $G$, there must be also an index $i \in \mathbb{N}$ such that $0 \le i \le k-1$ and $p_i \neq p_i'$. Index $i$ can be in three different positions relative to $j$ and $j'$:
	\begin{enumerate}
		\item if $i < j < j'$, then $(q_i, \dots, q_{j-1})(q_j, \dots, q_{j'-1})^\omega$ is an infinite and ultimately periodic walk in $G \otimes G$ and its projections are occurrences of an ultimately periodic, infinite and repeated string in $G$, since $\pi_1 (q_i) = p_i \neq p_i' = \pi_2(q_i)$, contradicting our hypothesis;
		\item if $j \le i \le j'$, then $(q_j, \dots, q_{j'-1})^\omega$ is a periodic walk in $G \otimes G$ and its projections are occurrences of a periodic, infinite and repeated string in $G$, a contradiction;
		\item if $j < j' < i$, then $(q_j, \dots, q_{j' - 1})(q_{j'}, \dots, q_i)$ is a walk in $G \otimes G$ with a cyclic prefix that can be pumped, so $(q_j, \dots, q_{j'-1})^m(q_{j'}, \dots, q_i)$ is a walk in $G \otimes G$ for each $m \ge 1$; the projections in $G$ of these strings are occurrences of repeated strings of the form $R^m S$, with $R = L' \big( (q_j, \dots, q_{j'-1}) \big)$ and $S = L' \big( (q_{j'}, \dots, q_i) \big)$,	since $\pi_1 (q_i) = p_i \neq p_i' = \pi_2 (q_i)$. \qedhere
	\end{enumerate}
\end{proof}

Point 3.\ of the above proof shows that the index where the projections of the distinct walks differ must occur after every cycle of the walk considered in $G \otimes G$, proving that any of these walks has a proper prefix of vertices of the form $(u,u)$ containing a cycle and this prefix ends in a vertex $(v,v) \in V_{\text{ndet}}'$. Note that $(u,u) = (v,v)$ can hold. We obtain the following result.

\begin{corollary}[Unbounded repeated strings]\label{cor:unbounded}
	If $G$ has no infinite repeated string, then $G$ has 
	repeated strings of unbounded length if and only if any $(u,u) \in V_{\text{cyc}}'$ reaches any $(v,v) \in V_{\text{ndet}}'$ with a path, and if $R$ is the spelling of a cycle starting and ending in $(u,u)$ and $S$ is the spelling of a path starting from $(u,u)$ and ending with $(v,v)$ and with an out-neighbor of $(v,v)$ with a sibling having the same label, then $R^m S$ is a repeated string for each $m \ge 1$.
\end{corollary}

We solve \LRSP\ on the general graph $G$ by combining \Cref{cor:infinite,cor:unbounded} and \Cref{thm:alg-dags}:

\begin{description}
	\item[\textbf{Infinite length repeats}] Check if any $q \in V'_{\text{diff}}$ reaches any $q' \in V'_{\text{cyc}}$ even with an empty path. If so, return ($i$) the string spelled by the path from $q$ to $q'$, ($ii$) the string spelled by any cycle starting from $q'$ and ($iii$) the symbol $\omega$. 
	\item[\textbf{Unbounded length repeats}] Check if any $(u,u) \in V'_{\text{cyc}}$ reaches any $(v,v) \in V'_{\text{ndet}}$ even with an empty path. If so, return ($i$) the string spelled by any cycle starting (and ending) at $(u,u)$, ($ii$) the symbol $+$ and ($iii$) the string spelled by the path from $(u,u)$ to an out-neighbor of $(v,v)$ with a sibling having the same label (since $(v,v) \in V'_{\text{ndet}}$).
	\item[\textbf{Finite length repeats}] Remove from $G \otimes G$ all vertices in $V'_{\text{cyc}}$ (obtaining a DAG), and proceed as in the algorithm for the DAG case from \Cref{sec:alg-dags} on this graph. 
\end{description}

\begin{theorem}
\label{thm:main-algorithm}
The above algorithm correctly solves \LRSP. Moreover, if $G \otimes G$ is given, it can be implemented to run in time linear in the size of $G \otimes G$. If not, it runs in time $O( \lvert E \rvert^2)$, where $E$ is the edge set of $G$.
\end{theorem}
\begin{proof}
If $G$ has infinite repeated strings, then \Cref{cor:infinite} guarantees the correctness of the first check. Otherwise, \Cref{cor:unbounded} guarantees the correctness of the second check.

Suppose now that both of these checks return false. First, since the first check failed, $V'_{\text{diff}} \cap V'_{\text{cyc}} = \emptyset$, because any vertex in $V'_{\text{diff}} \cap V'_{\text{cyc}}$ reaches itself with an empty path. Second,  from any $q \in V'_{\text{cyc}}$ no vertex $q' \in V'_{\text{diff}}$ is reachable (with a non-empty path, since $V'_{\text{diff}} \cap V'_{\text{cyc}} = \emptyset$). Indeed, suppose for a contradiction that $q' \in V'_{\text{diff}}$ is a vertex reached from $q$ with a shortest (non-empty) path $P$, and let $q^* \in V'$ be the vertex on this path right before $q'$. Since $P$ is shortest, then $q^* \notin V'_{\text{diff}}$, and thus $q^* \in V'_{\text{ndet}}$. However, this contradicts the assumption that the second check of the algorithm returned false.

Finally, since the two occurrences of a repeated string must pass through a vertex in $V'_{\text{diff}}$, and no vertex in $V'_{\text{diff}}$ is reached, or reaches a vertex in $V'_{\text{cyc}}$, then we can remove all vertices in $V'_{\text{cyc}}$ from $G \otimes G$, obtaining a DAG. In this DAG, as in \Cref{sec:alg-dags}, we look for the longest path passing through a vertex in $V'_{\text{diff}}$.

The SCCs of $G \otimes G$ and the sets $V'_{\text{cyc}}$, $V'_{\text{diff}}$, $V'_{\text{ndet}}$ can be computed in linear time in the size $G \otimes G$. Reachability between two sets of vertices of $G \otimes G$ (and a corresponding path) can also be implemented in linear time in the size of $G \otimes G$. The algorithm for the final DAG case runs in linear time in the size of $G \otimes G$, by \Cref{thm:alg-dags}.
\end{proof}

\section{Hardness}
\label{sec:hardness}

The NP-hardness of the \SMLG\ problem defined on path occurrences (implying the NP-hardness of \LCSP\ and of \MSP, also defined on path occurrences) was already observed in previous works such as~\cite{DBLP:conf/icalp/EquiGMT19}. We similarly observe that the same holds also for \LRSP.

\begin{observation}\label{thm:np-path}
	\LRSP\ defined on paths is NP-hard, even if we restrict alphabet $\Sigma$ to contain just a single character. This follows by reducing from the Hamiltonian Path problem on directed graphs. Given a graph $G$, create a graph $G'$ made up of two copies of $G$ and label all vertices with the same character. It easily holds that $G$ has a Hamiltonian path if and only if the length of the longest repeated string in $G'$ equals the number of vertices of $G$.
\end{observation}

As noted in the introduction, the quadratic lower bounds of \cite{DBLP:conf/icalp/EquiGMT19,DBLP:conf/sosa/GibneyHT21} for \SMLG\ imply the same quadratic lower bounds for \LCSP\ and \MSP, namely that the two problems cannot be solved in truly sub-quadratic time under the Orthogonal Vectors Hypothesis (that we discuss below in this section) and that the shaving from the quadratic-time complexity of arbitrarily high or high enough logarithmic factors would contradict other hardness conjectures. We now show a linear-time reduction from \SMLG\ to \LRSP\ on deterministic DAGs, which thus implies the same lower bounds for \LRSP\ as  in~\cite{DBLP:conf/icalp/EquiGMT19,DBLP:conf/sosa/GibneyHT21} (since they hold also for deterministic DAGs).

\begin{theorem}
\label{thm:SMLG-reduction}
Given a string $P \subseteq \Sigma^*$ and a $\Sigma$-labeled deterministic DAG $G = (V,E,L)$, there exists a $\Sigma'$-labeled DAG $G'$, with $\Sigma' = \Sigma \cup V$, having $O( \lvert V \rvert + \lvert E \rvert + \lvert P \rvert )$ vertices and edges, computable in linear time in the size of $P$ and $G$, and such that $P$ has an occurrence in $G$ if and only if the longest repeated string of $G'$ has length $\lvert V \rvert + \lvert P \rvert + 1$.
\end{theorem}
\begin{proof}
Given a deterministic graph $G = (V,E,L)$ and a pattern $P = P[1]\cdots P[m]$ labeled on alphabet $\Sigma$, with $V = \lbrace v_1, \dots, v_n \rbrace$ and $m,n > 0$, the reduction consists of transforming pattern $P$ into a labeled graph $G_P$, that is a path of $m$ vertices spelling $P$, and building a $\Sigma'$-labeled graph $G'$ with $\Sigma' = \Sigma \cup V$, assuming $V \cap \Sigma = \emptyset$. Graph $G'$, as seen in Figure~\ref{fig:smlgreduction}, contains $G$, $G_P$, and two copies $H_1$, $H_2$ of a simple gadget $H$ appropriately connected to them. Gadget $H$ is made of a path of $n$ vertices spelling $a^n$, with $a \in \Sigma$ chosen arbitrarily, ending in a level of $n$ vertices each labeled with a different vertex of $V$. In $H_1$ each of these final vertices is connected to the respective vertex of $G$ and in $H_2$ they are all connected to the source of $G_P$. The resulting graph $G'$ is a deterministic DAG made of two connected components each having exactly one source, and it is easy to see that the longest repeated string of $G'$ has length at most $\lvert V \rvert + \lvert P \rvert + 1$. Each repeated string of this maximum length has one occurrence per component, starting at the respective source. If $P$ has an occurrence $(u_1, \dots, u_m)$ in $G$, then $a^{n} u_1 P$ is a repeated string in $G'$ of maximum length. Conversely, every repeated string in $G'$ of maximum length $\lvert V \rvert + \lvert P \rvert + 1$ is of the form $a^n u_i P$, with $u_i \in V$, and its occurrence in the first component has as its suffix an occurrence of $P$ in $G$.

Graph $G'$ has $O( \lvert V \rvert + \lvert E \rvert + \lvert P \rvert )$ vertices and edges and its construction is straightforward, so the reduction takes linear-time in the size of the starting \SMLG\ instance.
\end{proof}

\begin{figure}[btp]
     \centering
     \includegraphics{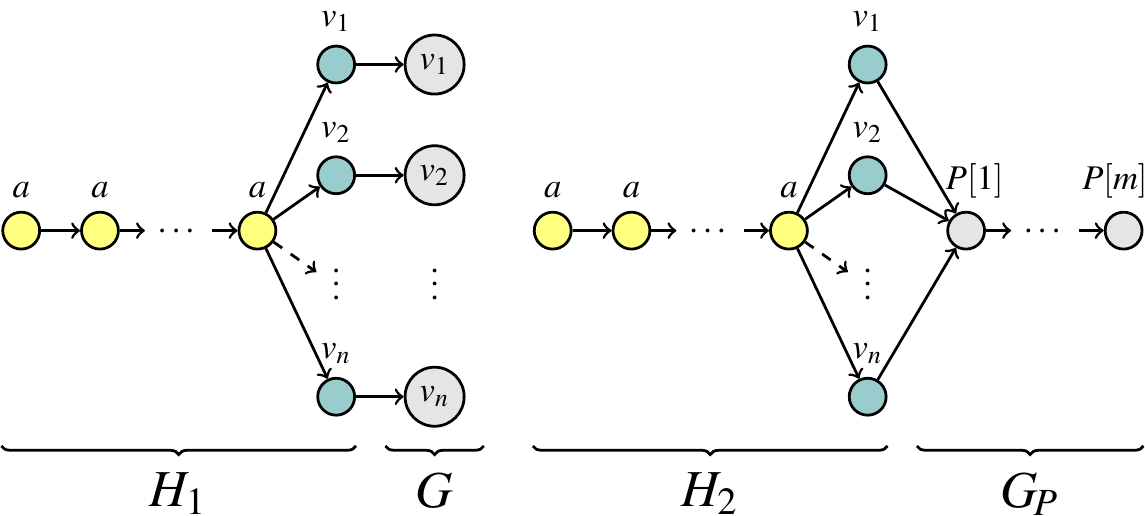}
     \caption{Scheme for the reduction of \SMLG\ to \LRSP: $H_1$ and $H_2$ are two copies of the same gadget made of $n+1$ levels and the edges of $G$ are not shown. Note that the reduction holds only if $G$ is a deterministic DAG.}\label{fig:smlgreduction}
\end{figure}

Two interesting aspects of \Cref{thm:SMLG-reduction} are as follows:
\begin{itemize}
    \item the reduction does not hold if $G$ contains cycles or if $G$ is a DAG with some non-deterministic vertices, because there could be infinite repeated strings in $G$ or the repeated strings of $G$ could be extended by gadget $H_1$;
    \item to the best of our knowledge, such a reduction does not exist when the problems are defined on strings.
\end{itemize}
Nevertheless, this reduction creates an instance of \LRSP\ with vertices of arbitrarily high in-degree and out-degree, and also increases the alphabet size by the number of edges of the graph. Therefore, in the rest of this section we give a direct reduction from the Orthogonal Vectors Problem (\OVP) to \LRSP, which will allow for both constant in- and out-degree, and binary alphabet.

In \OVP\ we are given two sets of binary vectors $A, B \subseteq \lbrace 0,1 \rbrace^d$, with $\lvert A \rvert = \lvert B \rvert = n$ and $d = \omega(\log n)$, and we need to determine whether there exist $a \in A$, $b \in B$ so that $a \cdot b = 0$, where $a \cdot b = \sum_{i=1}^{d} a[i] \cdot b[i]$. The Orthogonal Vectors Hypothesis (\OVH) states that no (randomized) algorithm can solve \OVP\ on instances of size $n$ in $O(n^{2-\varepsilon}\poly(d))$ time for constant $\varepsilon > 0$~\cite{Williams05}. Given an instance $A$, $B$ for \OVP, we will construct a DAG $G$ such that $A$ and $B$ contain a pair of orthogonal vectors if and only if the length of the longest repeated string in $G$ is of a certain value, to be introduced at the end of the reduction.

To start with, we use the alphabet $\Sigma = \lbrace 0,1,\mathtt{c} \rbrace$, where $\mathtt{c}$ is used to simplify the proofs. At the end, we will observe that all $\mathtt{c}$-labeled vertices can be relabeled with $0$.
We start by building two types of gadgets:
\begin{itemize}
	\item for each $a = (a[1], \dots, a[d]) \in A$, graph $G_a$ is a path consisting of a starting \texttt{c}-labeled vertex followed by $d$ vertices, where the $i$-th vertex is labeled with $a[i]$ (Figure~\ref{fig:gadgets}, left);

	\item for each $b = (b[1], \dots, b[d]) \in B$, graph $\overline{G}_b$ is a DAG with $d+1$ levels such that: (i) the zeroth level consists of a single \texttt{c}-labeled source vertex; (ii) the $i$-th level has both a $0$-labeled vertex and a $1$-labeled vertex if $b[i] = 0$, otherwise (if $b[i] = 1$) it just has a $0$-labeled vertex. All vertices in each level have edges going to all vertices in the next level, and there are no edges between vertices of the same level (Figure~\ref{fig:gadgets}, right). This is the same type of gadget used also in~\cite{DBLP:conf/icalp/EquiGMT19}.

\begin{figure}[btp]
	\centering
	\includegraphics[scale=1]{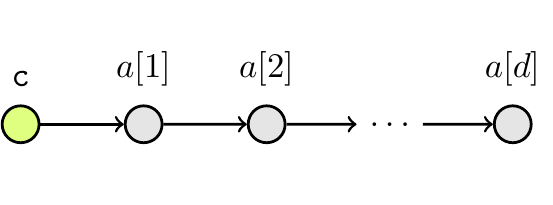}%
	\qquad
	\includegraphics[scale=1]{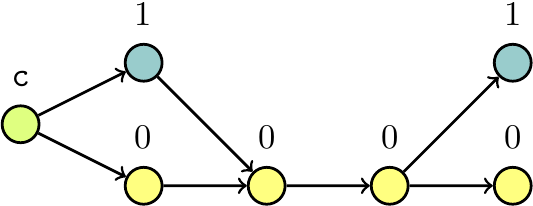}
	\caption{Gadget $G_a$ (left), with $a = (a[1], \dots, a[d]) \in A$, and 
	 $\overline{G}_b$ (right), with $b = ( 0, 1, 1, 0 ) \in B$. Note that since $b[2] = b[3] = 1$, the second and third levels of $\overline{G}_b$ only have one $0$-labeled vertex.}\label{fig:gadgets}
\end{figure}
\end{itemize}

To build up the intuition, take $a \in A$ and $b \in B$ and observe, similarly as in~\cite{DBLP:conf/icalp/EquiGMT19}, that the string spelled by $G_a$ has an occurrence in $\overline{G}_b$ if and only if $a$ and $b$ are orthogonal. Thus, the graph made up of a copy of $G_a$ and a copy of $\overline{G}_b$ has a longest repeated string of length $d + 1$ if and only if $a$ and $b$ are orthogonal. However, we cannot put together all gadgets $G_a$ and $\overline{G}_b$ as separate components of the same graph, because such a simple construction cannot restrict the location of occurrences of the longest repeated string. Intuitively, we need the longest repeated string to have one occurrence in the part of the graph corresponding to the $G_a$ gadgets, and one occurrence in the part of the graph corresponding to the $\overline{G}_b$ gadgets. We achieve this by (i) building a tree structure on top of the $\overline{G}_b$ gadgets that assigns to each gadget its own unique prefix; and by (ii) building a ``universal'' structure on top of gadgets $G_a$ to make them reachable by reading any possible prefix added to gadgets $\overline{G}_b$. More specifically, we introduce the following two gadgets with $\lceil \log_2 n \rceil + 1 = k + 1$ levels:

\begin{itemize}
	\item gadget $T$, seen in Figure~\ref{fig:gadgetstu} (left), is a complete binary tree of height $k + 1$ and with $2^{k} \ge n$ leaves, in which the root is \texttt{c}-labeled, all left children are $0$-labeled and all right children are $1$-labeled; trivially, any root-to-leaf path in such a tree has a different label; 
	\item a ``universal'' DAG $U$ with a \texttt{c}-labeled source followed by $k$ levels of vertices where each level has two vertices, labeled with $0$ and $1$, and each vertex in a level is connected to the vertices of the next level, as can be seen in Figure~\ref{fig:gadgetstu} (right).
\end{itemize}
\begin{figure}[btp]
	\centering
	\includegraphics[scale=1]{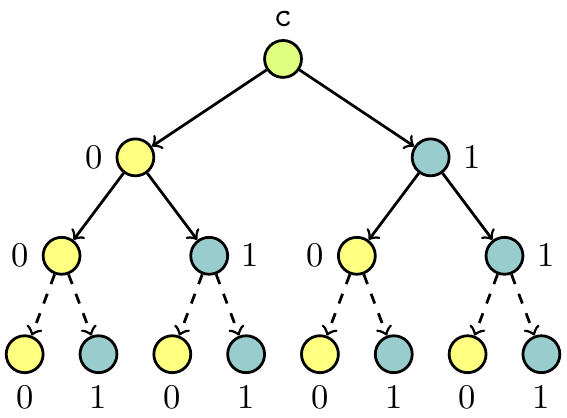}%
	\qquad\qquad%
	\includegraphics[scale=1]{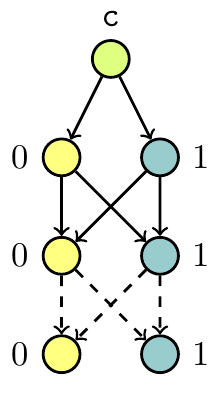}
	\caption{Gadget $T$ (left), a complete binary tree with $k+1$ levels, and gadget $U$ (right), reading every possible string of length $k+1$ that can be read in $T$.}\label{fig:gadgetstu}
\end{figure}

Our gadgets can be arranged in a \emph{non-deterministic} DAG as seen in Figure~\ref{fig:reduction1} (left): the two sinks of gadget $U$ are connected to each source of gadgets $G_a$, with $a \in A$, and each leaf of the tree gadget $T$ is connected to the source of a different gadget $\overline{G}_b$, with $b \in B$; if $n$ is not a power of two, some leaves of gadget $T$ can be left without any out-neighbors. To have a \emph{deterministic} DAG, we can further merge  all gadgets $G_a$ in a keyword tree (trie) $K_{G_{a_1}, \dots, G_{a_n}}$ (see Figure~\ref{fig:reduction1}, right), so the set of strings spelled by the entire graph is unchanged (and the leaves of the keyword tree remain all distinct since all vectors in $A$ are distinct). Note that the longest path in this graph has length $k + d + 1$, with $k = \lceil \log_2 n \rceil$.

\begin{figure}[t]
	\centering
	\includegraphics{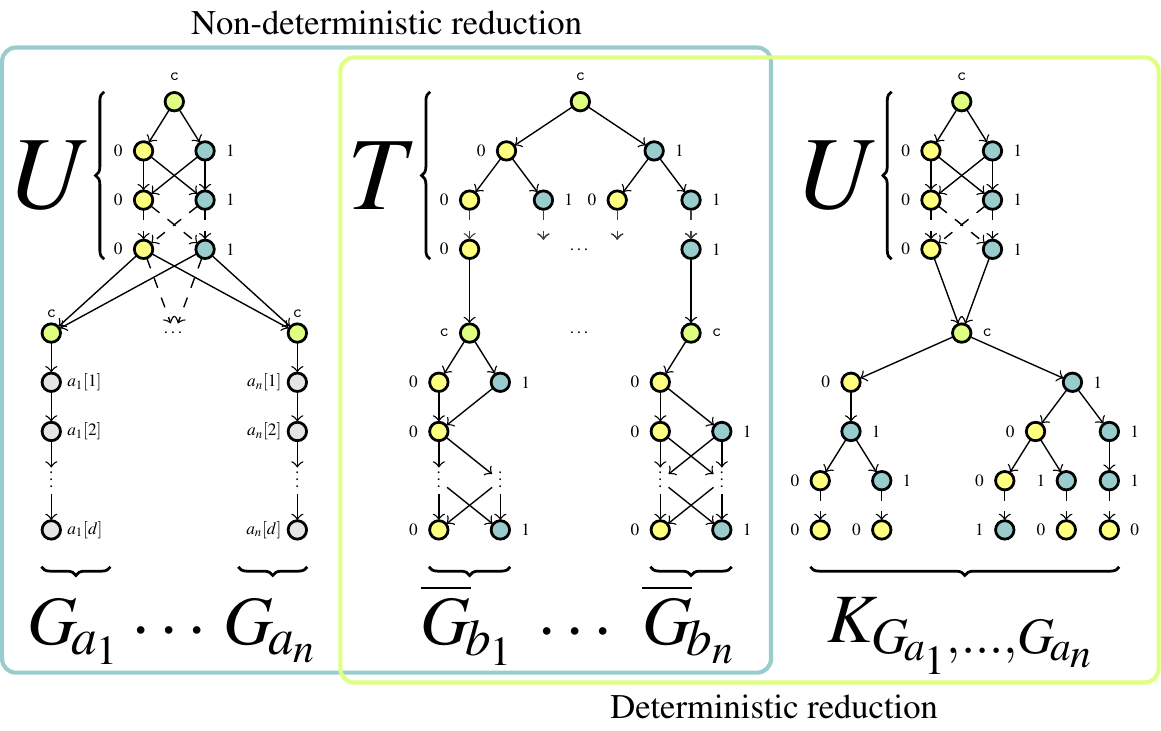}
	\caption{First scheme (left) for the OV reduction, made of two subgraphs $G_A$ (left) and $G_B$ (right): $G_A$ is non-deterministic; second scheme (right) for the OV reduction: $K_{G_{a_1}, \dots, G_{a_n}}$ is the keyword tree (trie) of gadgets $G_a$, $a \in A$.}\label{fig:reduction1}
\end{figure}

\begin{lemma}\label{lem:reduction2}
	For an instance $A$ and $B$ for \OVP, the deterministic DAG $G$ built as in Figure~\ref{fig:reduction1} has a repeated string of length $k + d + 2$, with $k = \lceil \log_2 n \rceil$, if and only if there exist $a \in A$, $b \in B$ orthogonal.
\end{lemma}
\begin{proof}
    By construction, the longest paths in $G$ have  length $k + d + 1$ and thus they spell strings of length $k + d + 2$. 
Moreover, all these longest strings are of the form $\mathtt{c}S_1\mathtt{c}S_2$, with $S_1 \in \lbrace 0, 1 \rbrace^k$ and $S_2 \in \lbrace 0, 1 \rbrace^d$.

	($\Rightarrow$) If there is a repeated string $\mathtt{c}S_1\mathtt{c}S_2$ of length $k + d + 2$ then there must be exactly two occurrences of it, one in $G_A$ and one in $G_B$, since the graph is deterministic and has only two sources from which the longest strings can be read. This implies the existence of $a' \in A$, $b' \in B$ such that $L(a') = S_2$ and $a' \cdot b' = 0$, due to trie $K_{G_{a_1}, \dots, G_{a_n}}$ and due to the properties of the longest strings of gadgets $G_{a'}$ and $\overline{G}_{b'}$.

	($\Leftarrow$) Given $a' \in A$, $b' \in B$ such that $a' \cdot b' = 0$, let $\mathtt{c}S_1$, with $S_1 \in \lbrace 0,1 \rbrace^k$, be the string corresponding to the unique path from the \texttt{c}-labeled source of $G_B$ to gadget $\overline{G}_{b'}$. Then string $\mathtt{c}S_1\mathtt{c}S_2$, with $S_2 \in \lbrace 0,1 \rbrace^d$ the linearization of vector $a'$, has two occurrences in $G$, one in $G_A$, passing through gadgets $U$ and $K_{G_{a_1}, \dots, G_{a_n}}$, and one in $G_B$, passing through $T$ and $\overline{G}_{b'}$.
\end{proof}
	To make the alphabet binary, it is easy to see that it suffices to relabel all \texttt{c}-labeled vertices with $0$.
	This proves that, under \OVH, there can be no truly sub-quadratic time algorithm.
\begin{theorem}
\label{thm:LRSPcomplexity}
	If \OVH\ holds, then for no $\varepsilon > 0$ there is a $O \big( \lvert V \rvert^{2 - \varepsilon} \big)$-time or $O \big( \lvert E \rvert^{2 - \varepsilon} \big)$-time algorithm for \LRSP, even when restricted to deterministic DAGs, labeled with a binary alphabet, in which both the maximum in-degree and out-degree of any vertex are at most 2.
\end{theorem}
\begin{proof}
    It is easy to check that the maximum in-degree, and out-degree of any vertex of the graph in the reduction of \Cref{fig:reduction1} is at most 2. Lemma~\ref{lem:reduction2} guarantees the correctness of the reduction, so it remains to analyze its complexity. The resulting graph $G$ has $O(nd)$ vertices and $O(nd)$ edges and can be constructed in $O(nd)$ time, since the keyword tree can be constructed in time linear in the size of its inputs. Thus, if \LRSP\ has an $O \big( \lvert V \rvert^{2 - \varepsilon} \big)$-time or an $O \big( \lvert E \rvert^{2 - \varepsilon} \big)$-time algorithm for some $\varepsilon > 0$, \OVP\ has an $O\big( (nd)^{2 - \varepsilon} \big)$-time algorithm, contradicting \OVH.
\end{proof}

Our \OVH\ reduction for \LRSP\ immediately proves an \OVH\ reduction also for \LCSP\ (by taking the two components of $G$ as input graphs $G_1$ and $G_2$ for \LCSP). 

It also provides a quadratic lower bound for an apparently simpler version of \MSP, which on two paths (i.e.~strings) can be solved in a trivial manner. Let \MSP$^\ast$ be defined as $\MSP$, with the difference that we are given a \emph{single} vertex $v_1$ of $G_1$, a \emph{single} vertex $v_2$ of $G_2$, and we need to compute the length of the longest string having an occurrence in $G_1$ starting at $v_1$ and an occurrence in $G_2$ starting at $v_2$. To obtain the \OVH\ reduction, it can be easily checked that it suffices to take as $G_1$ the subgraph of $G$ built from $B$ (\Cref{fig:reduction1}, middle) with $v_1$ being its source vertex, and as $G_2$ the graph build from $A$ (\Cref{fig:reduction1}, right) with $v_2$ being its source vertex.

\begin{corollary}
    \label{cor:MSP*}
    If \OVH\ holds, then for no $\varepsilon > 0$ there is a $O \big( ( \lvert V_1 \rvert \cdot \lvert V_2 \rvert )^{1 - \varepsilon} \big)$-time or $O \big( ( \lvert E_1 \rvert \cdot \lvert E_2 \rvert )^{1 - \varepsilon} \big)$-time algorithm for \MSP$^\ast$, even when both input graphs are deterministic DAGs, labeled with a binary alphabet, in which the maximum in-degree and out-degree of any vertex is at most 2.
\end{corollary}

Even though the above result about $\MSP^\ast$ holds also for deterministic DAGs, its hardness stems from the fact that we do not know which path in $G_1$ to match with a path in $G_2$ in order to maximise their length. However, if $G_2$ is just a path, then the problem is solvable in linear time.

\section{\LRSP\ on Undirected Graphs}
\label{sec:undirected}

In this section we study \LRSP\ on undirected graphs, namely when the occurrences of a repeated string are allowed to use an undirected edge in any of its two directions. Even if these results are easier than the results on directed graphs, they complete the complexity picture of \LRSP. We will study the variants of \LRSP\ on undirected paths and undirected walks, and consider the same classes of undirected graphs: paths, trees\footnote{By \emph{undirected tree} we mean an unrooted tree, where an occurrence can use an undirected edge in either direction.} and general graphs\footnote{For simplicity, we assume that self-loops in undirected graphs are not present, even if they do not change the results.}.

\begin{theorem}
\LRSP\ on undirected graphs and defined on path occurrences can be solved as follows:
\begin{enumerate}
		\item On an undirected graph $G$ that is a path, \LRSP\ can be solved in linear time. 
		\item On an undirected graph $G$ that is a tree, \LRSP\ can be solved in quadratic time. 
		\item On general undirected graphs, \LRSP\ is NP-complete, since the same reduction as in \Cref{thm:np-path} works also for undirected graphs.
	\end{enumerate}
\end{theorem}

\begin{proof}
    For 1., note that the occurrences of a repeated string are obtained by either moving only forward or only backward (because \LRSP\ is defined on paths, and thus occurrences cannot repeat vertices). Thus, if $T$ is the spelling of $G$ from one end to another, we can reduce \LRSP\ on $G$ to finding the longest repeated substring of the text $T\$T^{-1}$, where $\$$ is a new separator character, and $T^{-1}$ is $T$ reversed. Thus, we can solve \LRSP\ on $G$ in linear time.
    
    For 2., we show that \LRSP\ on an undirected tree with $n$ vertices, defined with path occurrences, can be reduced to \LRSP\ on a directed tree with $O(n^2)$ vertices (where there is no distinction between path and walk occurrences). Indeed, let $v_1,\dots,v_n$ be the vertices of a $\Sigma$-labeled undirected tree $T$. For each $i \in \{1,\dots,n\}$, construct the directed tree $T_{v_i}$ by setting $v_i$ as root and orienting all edges away from $v_i$. Also, let $(u_1,\dots,u_n)$ be a directed path of $n$ new vertices. Construct the directed tree $T'$, rooted at $u_1$, by combining the path $(u_1,\dots,u_n)$ with trees $T_{v_1},\dots,T_{v_n}$, adding the directed edges $(u_n,v_i)$, for all $i \in \{1,\dots,n\}$. The vertices of each $T_{v_i}$ are labeled as in $T$, and $u_1,\dots,u_n$ are labeled with a new character $\$ \notin \Sigma$. Clearly, the number of vertices of $T'$ is $n^2 + n$. We claim that $T$ has a longest repeated string of length $\ell$ if and only if $T'$ has a longest repeated string of length $n+\ell$, spelled by a path starting with $(u_1,\dots,u_n)$. This is proved by combining the following remarks:
    \begin{itemize}
        \item All occurrences of any repeated string of $T'$ longer than $n$ must start in the same vertex $u_i$, with $i \in \lbrace 1, \dots, n \rbrace$, since $\$ \notin \Sigma$. Moreover, if a repeated string of \emph{maximum} length in $T'$ has length greater than $n$ then its occurrences start at $u_1$.
        \item Consider two distinct occurrences of a longest repeated string in $T'$ of length $n + \ell$, with $\ell \ge 1$, both starting from $u_1$. By construction of $T'$, their suffixes of length $\ell$ correspond to two distinct occurrences of a string of length $\ell$ in $T$. Vice versa, given two distinct occurrences of a repeated string $w$ in $T$, there are two corresponding distinct occurrences of $\$^n w$ in $T'$ of length $n + \ell$, both starting from $u_1$.
    \end{itemize}
    Thus, the longest repeated strings in $T$ correspond to the longest repeated strings in $T'$ and vice versa (if there are no repeated strings in $T$ then the repeated strings of $T'$ have length lesser than $n$), so we can apply the linear-time solutions based on the suffix tree of a tree or the XBWT to obtain a globally quadratic-time algorithm.

    For 3., observe that the same reduction as in \Cref{thm:np-path} works also for undirected graphs (since the Hamiltonian path problem is NP-hard also on undirected graphs).
\end{proof}

For \LRSP\ defined on walk occurrences, observe first that we can replace each undirected edge with a pair of edges oriented in opposite directions. Thus, we can solve the problem in quadratic time, using the algorithm from \Cref{sec:generalgraphs} for directed graphs. However, we show that \LRSP\ can be solved in \emph{linear} time on general undirected graphs, using the following lemma, greatly simplifying the problem.

\begin{lemma}\label{lem:undirected}
	Given a $\Sigma$-labeled undirected graph $G = (V,E,L)$, $G$ has a repeated string of length at least 2 if and only if $G$ has an infinite repeated string.
\end{lemma}
\begin{proof}
	Let $p = (p_0, p_1, \dots, p_{k-1})$ and $p' = (p_0', p_1', \dots, p_{k-1}')$ be distinct occurrences of a string, such that $p_i \neq p_i'$ for some $0 \le i \le k - 1$. We can build an infinite repeated string with just two pair of adjacent vertices visited by $p$ and $p'$:
	\begin{itemize}
		\item if $i = 0$ then
			$(p_0, p_1)^\omega$, $(p_0', p_1')^\omega$
			are distinct occurrences of $(L(p_0) L(p_{1}))^\omega$;
		\item if $i >0$ then
			$(p_{i-1}, p_i)^\omega$, $(p_{i-1}', p_i')^\omega$
			are distinct occurrences of $(L(p_{i-1}) L(p_i))^\omega$.\qedhere
	\end{itemize}
\end{proof}

\Cref{lem:undirected} implies we can just check if there is a repeated string of length $2$ (in which case there is an infinite repeated string) and, if not, of length $1$. These two checks can be done in linear time, provided that the vertex set and the edge set of $G$ are already ordered in lexicographical ordering:
    there is a repeated string of length $2$ if and only if there are two distinct edges with matching edges, and
    there is a repeated string of length $1$ if and only if there are two distinct vertices with the same label.

\section{Conclusions and Future Work}\label{sec:conclusions}
In this paper we introduced the labeled direct product graph as a straightforward algorithmic tool, since it naturally encodes all pairs of walks in the original graphs having matching labels. Through simple applications, we developed optimal and predictable algorithms for existing problems on labeled graphs---string matching in labeled graphs (\SMLG) and longest common substring (\LCSP)---and for extensions of string problems that we introduced---matching statistics (\MSP) and longest repeated string (\LRSP). For \SMLG\ and \LCSP\ this resulted in more efficient algorithms than the existing quadratic-time ones, since the product graph excludes all pairs of mismatching vertices and edges.

Regarding complexity, we extended the existing conditional quadratic-time lower bounds for \SMLG\ of \cite{DBLP:conf/icalp/EquiGMT19,DBLP:conf/sosa/GibneyHT21} to \LRSP\ with a linear-time reduction, if the input graph of \SMLG\ is a deterministic DAG. Since the \SMLG\ lower bounds trivially hold for \LCSP\ and \MSP, this means that our algorithms (and the existing one for \LCSP) are conditionally optimal. Moreover, we designed a single, more efficient reduction from the Orthogonal Vectors Problem (\OVP) to \LCSP, \MSP\ and \LRSP\ proving that the three problems cannot be solved in truly sub-quadratic time under the Orthogonal Vectors Hypothesis (\OVH), even if the graphs in input are acyclic, deterministic (i.e.\ every vertex has at most one $a$-labeled out-neighbor, for every $a \in \Sigma$), labeled from a binary alphabet, and such that the maximum in-degree and out-degree of any vertex are at most 2. An interesting aspect of these results is that there is no known reduction of \SMLG\ to \LRSP\ when the problems are defined on strings and that the \OVP\ reduction holds also for the modification of \MSP\ trying to match the walks starting from just two vertices, even when the graphs have the same restrictions as before.

Our algorithms are based on linear-time analyses of the labeled direct product graph corresponding to each problem, so we spent some effort in studying its construction. Indeed, if the sets of vertices and edges of the graphs are sorted following the lexicographical order, then the construction of the product graph takes time and space linear in its size, thus under the standard assumption to work with an integer alphabet our algorithms globally reach this time and space complexity. This also means that the size of the labeled direct product graph is a tighter complexity upper bound for \SMLG, \LCSP, \MSP\ and \LRSP. Plus, the size of the product graph can be precomputed in time linear in the size of the input graphs, making it possible to report the run time of our algorithms before their computation. If the alphabet has constant size, there is no need to store the edges of the product graph, whereas if the alphabet is an integer one then the choice not to store the edges of the graph is a version of the \textsf{SetIntersection} problem, leading to a space and time trade-off.

Finally, we presented a complete complexity picture of \LCSP\ and \LRSP\ on different classes of directed graphs and we did the same for \LRSP\ on undirected graphs. The only open case is \LRSP\ defined on path occurrences (i.e.\ when there are no repeated vertices) in undirected trees, for which we obtained only a quadratic-time algorithm in \Cref{sec:undirected}, with no matching lower bound. Since the number of different paths of an undirected tree is only quadratic, we believe this problem cannot encode an \OVP\ instance.
Thus, we pose the open problem of finding a linear-time algorithm for this variant.

Recall that in the introduction we encoded a labeled graph as an NFA and we argued that \SMLG, \LCSP\ and \LRSP\ (where we focus on finite strings) are special cases of similarly defined problems for finite-state automata (over finite words). The quadratic-time conditional lower bounds automatically carry over to these problems, and as the classical quadratic-size construction of an NFA recognizing each and every word accepted by two input NFAs solves \LCSP, we deem that there is a quadratic-size NFA encoding all ambiguous words of any input NFA thus solving \LRSP, and we leave this for future work.

We also leave as future work to find more problems on labeled graphs solved by the labeled direct product graph, or that can be tackled with the same general strategy of precomputing a data structure to globally obtain time savings during the actual computation.

\bibliography{ms}{}
\bibliographystyle{plainurl}
\end{document}

%% file: config.tex
\theoremstyle{plain}%
\newtheorem{theorem}{Theorem}

\newtheorem{problem}{Problem}
\newtheorem{corollary}{Corollary}

\theoremstyle{remark}%
\newtheorem{remark}{Remark}%

\theoremstyle{definition}%
\newtheorem{lemma}{Lemma}
\newtheorem{definition}{Definition}%
\newtheorem{observation}{Observation}
\crefname{observation}{observation}{observations}
\crefname{observation}{Observation}{Observations}

\DeclareMathOperator{\poly}{poly}

\DeclareMathOperator{\MS}{MS}

\newcommand{\LRSP}{\textsf{LRSP}}
\newcommand{\SMLG}{\textsf{SMLG}}
\newcommand{\OVP}{\textsf{OVP}}
\newcommand{\OVH}{\textsf{OVH}}
\newcommand{\LCSP}{\textsf{LCSP}}
\newcommand{\MSP}{\textsf{MSP}}